\documentclass[11pt]{article}
\pdfoutput=1
\usepackage{rpmacros}
\usepackage{amsmath,amsthm,amssymb}
\usepackage{mathtools}
\usepackage{fullpage}
\usepackage{mathrsfs}
\usepackage{stmaryrd}
\usepackage{todonotes}
\usepackage{graphicx}
\usepackage{xspace}
\usepackage{tikz}
\usepackage{tikz-cd}
\usepackage{framed}
\usepackage{float}
\usepackage{xparse}
\usepackage{thm-restate}
\usepackage{setspace}
\usepackage{gitinfo}
\usepackage[colorlinks,pagebackref]{hyperref}
\hypersetup{
  linkcolor=[rgb]{0,0,0.4},
  citecolor=[rgb]{0, 0.4, 0},
  urlcolor=[rgb]{0.6, 0, 0}
}
\newcommand{\ehref}[1]{\href{mailto:#1}{#1}}

\usepackage{amsthm}
\usepackage{thmtools,thm-restate}
\usepackage[nameinlink,capitalise]{cleveref}

\numberwithin{equation}{section}
\declaretheoremstyle[bodyfont=\it,qed=\qedsymbol]{noproofstyle}

\declaretheorem[name=Observation,numbered=no]{observation*}

\declaretheorem[numberlike=equation]{theorem}

\declaretheorem[name=Theorem,numbered=no]{theorem*}

\declaretheorem[numberlike=equation]{lemma}
\declaretheorem[name=Lemma,numbered=no]{lemma*}

\declaretheorem[name=Corollary,numbered=no]{corollary*}

\declaretheorem[name=Proposition,numbered=no]{proposition*}

\declaretheorem[name=Claim,numbered=no]{claim*}

\declaretheorem[name=Conjecture,numbered=no]{conjecture*}

\declaretheorem[name=Question,numbered=no]{question*}

\declaretheoremstyle[bodyfont=\it,qed=$\lozenge$]{defstyle} 

\declaretheorem[numberlike=equation,style=defstyle]{definition}
\declaretheorem[unnumbered,name=Definition,style=defstyle]{definition*}

\declaretheorem[unnumbered,name=Example,style=defstyle]{example*}

\declaretheorem[unnumbered,name=Notation=defstyle]{notation*}

\declaretheorem[unnumbered,name=Construction,style=defstyle]{construction*}

\declaretheorem[unnumbered,name=Remark,style=defstyle]{remark*}

\allowdisplaybreaks

\newcommand{\gitinfonote}{git info:~\gitAbbrevHash\;,\;(\gitAuthorIsoDate)\; \;\gitVtag \; FRS codes = rmeoendcsoodleos }

\newif\ifdraft
\draftfalse

\ifdraft
\newcommand{\gitinfonotecolour}{Gray}
\newcommand{\RPnote}[1]{\textcolor{BrickRed}{$\langle$ RP: #1 $\rangle$}}
\newcommand{\AGnote}[1]{\textcolor{Blue}{$\langle$ AG: #1 $\rangle$}}
\newcommand{\phnote}[1]{\textcolor{purple}{$\langle$ sriprahladuvacha: #1 $\rangle$}}
\else
\newcommand{\gitinfonotecolour}{white}
\newcommand{\RPnote}[1]{}
\newcommand{\AGnote}[1]{}
\newcommand{\phnote}[1]{}
\fi

\newcommand{\FRS}{\operatorname{FRS}}
\renewcommand{\epsilon}{\varepsilon}

\title{An exposition of recent list-size bounds of FRS Codes}
\author{Abhibhav Garg\thanks{University of Waterloo, Canada. \ehref{abhibhav.garg@uwaterloo.ca}. Work done when the author was visiting TIFR.}
\and
Prahladh Harsha\thanks{Tata Institute of Fundamental Research, Mumbai, India. \href{mailto:prahladh@tifr.res.in}{\{prahladh},\href{mailto:mrinal@tifr.res.in}{mrinal},\href{mailto:ramprasad@tifr.res.in}{ramprasad},\href{mailto:ashutosh.shankar@tifr.res.in}{ashutosh.shankar\}@tifr.res.in}.}
\and 
Mrinal Kumar\footnotemark[2]
\and
Ramprasad Saptharishi\footnotemark[2]
\and
Ashutosh Shankar\footnotemark[2]}
\date{}

\begin{document}
\onehalfspacing
\maketitle

\begin{abstract}

In the last year, there have been some remarkable improvements in the combinatorial list-size bounds of Folded Reed Solomon codes and multiplicity codes. Starting from the work on Kopparty, Ron-Zewi, Saraf and Wootters \cite{KoppartyRSW2023} (and subsequent simplifications due to Tamo \cite{Tamo2024}), we have had dramatic improvements in the list-size bounds of FRS codes\footnote{While all the results in this note refer to FRS codes, they extend to all affine FRS codes, which includes multiplicity codes and additive-FRS codes.} due to Srivastava \cite{Srivastava2025} and Chen \& Zhang \cite{ChenZ2025}. In this note, we give a short exposition of these three results (Tamo, Srivastava and Chen-Zhang). 
\end{abstract}

\section{Introduction}

We start by defining Folded Reed Solomon (FRS) codes and list decoding capacity. Folded Reed-Solomon codes were introduced by Krachkovsky~\cite{Krachkovsky2003}, and then re-discovered by Guruswami and Rudra \cite{GuruswamiR2008} in the context of list-decoding.
Let $\bF_{q}$ be a finite field of $q$ elements, with $q > k$. 
\begin{definition}[folded Reed-Solomon codes (FRS) \cite{Krachkovsky2003,GuruswamiR2008}]
  Let $S = \set{\alpha_1,\ldots, \alpha_n}$ be a set of $n$ distinct elements in $\bF_q$ and let $\gamma$ be a generator of $\bF_q^*$. 
  The \emph{folded Reed-Solomon code} with parameters $(k, S, s)$ is defined via the following map:
  \begin{align*}
    \FRS_{k,s}\colon \bF_q[x]^{<k} &\rightarrow (\bF_q^s)^n\\
    f(x) & \mapsto \inparen{\insquare{\begin{array}{c} f(\alpha_1) \\ f(\gamma \alpha_1)\\ \vdots\\ f(\gamma^{s-1} \alpha_1) \end{array}}, \ldots, \insquare{\begin{array}{c} f(\alpha_n) \\ f(\gamma \alpha_n)\\ \vdots\\ f(\gamma^{s-1} \alpha_n) \;\end{array}}}\;.
  \end{align*}
  The parameter $s$ is also referred to as the \emph{folding parameter} of the $\FRS$ code. 

  The \emph{rate} of the above code shall be donoted by $R$, with $R := k / ns$, and it is known that the fractional distance of the code is $1 - R$.
\end{definition}
For the rest of the article, the set $S = \set{\alpha_1,\dots, \alpha_n}$ will be fixed and we will just refer to the FRS code as $\FRS_{k,s}$ code. We overload notation and use the same symbol to refer to both the polynomials (which correspond to messages) and their encodings under the above map.

\paragraph{List-decodability:} The notion of distance between codewords would be the standard Hamming distance. 

\begin{definition}[Hamming balls] 
  \label{defn:hamming-ball}
  For any point $y \in \Sigma^n$ for some alphabet $\Sigma$, we denote the \emph{Hamming ball of fractional radius $\rho$ around $y$} by $B(y, \rho)$ defined as
  \[
  B(y,\rho) := \setdef{x \in \Sigma^n}{\abs{\setdef{i\in [n]}{x_i \neq y_i}} < \rho n} \;.\qedhere
  \]
\end{definition}

The primary objective in list-decodability is to understand up to what radius do Hamming balls have ``few'' codewords. 

\begin{definition}[List-decodability]
  \label{defn:list-decodability}
  A code $\cC \subseteq \Sigma^n$ is said to be $(\rho, L)$ list-decodable if for every $y \in \Sigma^n$ we have
  \[
  \abs{C \cap B(y, \rho)} \leq L\;. \qedhere
  \]
\end{definition}
Folded Reed-Solomon codes were shown to achieve list-decoding capacity by Guruswami and Rudra~\cite{GuruswamiR2008}. That is, the set of codewords in a ball of radius $1 - R - \varepsilon$ around any point in the code space is small. 

Guruswami and Wang~\cite{GuruswamiW2013} re-proved this result in the following specific way: for FRS codes with folding parameter $O(1 / \varepsilon^{2})$, for any point $y$ in the code space, they show the existence of a linear subspace $\cA \subset \bF_{q}\bs{x}^{<k}$ with $\dim \cA = O(1 / \varepsilon)$ such that every codeword in the ball $B(y, 1 - R - \varepsilon)$ is the encoding of a polynomial in $\cA$.
This implies that the list size is at most $q^{O(1 / \varepsilon)}$.
Their proof of the existence of the subspace is algorithmic. We state this as a lemma below.

\begin{lemma}[Guruswami-Wang \cite{GuruswamiW2013}]
  \label{lemma:Guruswami}
  Let $y \in \br{\bF_q^s}^n$ be a received word for a $\FRS_{k,s}$ code of rate $R = k/ns$. Then, for $\rho = 1 - R - \epsilon$, if $s = \Omega(1/\epsilon^2)$, there is an affine space $\cA$ of dimension $O(1/\epsilon)$ that contains all codewords in $B(y,\rho) \cap C$. Furthermore, an affine basis for $\cA$ can be obtained in time polynomial in $n,\log q, 1/\epsilon$ given the received word $y$. 
\end{lemma}
Subsequently, Kopparty, Saraf, Ron-Zewi and Wootters~\cite{KoppartyRSW2023} showed that the upper bound on the list size at radius $1 - R - \varepsilon$ for such codes can be improved from polynomial $q^{O(1 / \varepsilon)}$ to constant $(1 / \varepsilon)^{O(1 / \varepsilon)}$.
A cleaner analysis of this upper bound was given by Tamo~\cite{Tamo2024}. These proofs were also algorithmic: they build on the previous result by taking the subspace as input and ``pruning'' the list size in a randomized fashion. In particular, if $\cL$ denotes the list of codewords in the ball $B(y,1-R-\epsilon)$, i.e., $\cL := B(y, 1 - R - \epsilon) \cap \FRS_{k,s}$, then the KRSW/Tamo improvement can be written as follows, a further simplified proof of which is presented in \cref{sec:KRSW}.

\begin{theorem}[\cite{KoppartyRSW2023,Tamo2024}]
  The size of $\cL$ is upper-bounded by $(1/\epsilon)^{O(1/\epsilon)}$. 
\end{theorem}

How small can the list-size bound be? Shangguan and Tamo \cite{ShangguanT2023} generalized the classical Singleton bound to show that any code with rate $R$ that is $(1-R-\epsilon,L)$ list-decodable satisfies $L \geq \frac{1-R-\epsilon}{\epsilon}$. Very recently, Srivastava \cite{Srivastava2025} and Chen \& Zhang \cite{ChenZ2025} gave dramatic improvements on this list-size to $O(1/\epsilon)$ almost matching the generalized Singleton bound up to a constant multiplicative factor.

\begin{theorem}[\cite{Srivastava2025}]
  The size of $\cL$ is upper-bounded by $O(1/\epsilon^2)$. 
\end{theorem}

\begin{theorem}[\cite{ChenZ2025}]
  The size of $\cL$ is upper-bounded by $O(1/\epsilon)$. 
\end{theorem}

We will give simplified proofs of these improvements in \cref{sec:Sri,sec:CZ}. It is to be noted that these proofs are combinatorial and algorithmizing them (efficiently; in time nearly-linear or even polynomial in the list size) remains open. The results of KRSW/Tamo also extend to list-recovery. However, as Chen-Zhang observe the dramatic improvements on list-size bounds for list-decoding FRS codes obtained by Srivastava and Chen \& Zhang do not extend to list-recovery of FRS codes. We (re-)present the Chen \& Zhang counterexample in \cref{sec:recovery}.

\paragraph{Agreement graphs:} A key ingredient we will be using in the proofs of these improvements is the notion of an agreement graph, which we define below.
\begin{definition}[Agreement graph]
  \label{defn:agreement-graph} For any code $\cC \subseteq \Sigma^n$, message $y \in \Sigma^{n}$, and a set of distinct codewords $f_1, \dots, f_m \in \cC$, the \emph{agreement graph} $G(\set{f_1,\dots f_m}, y)$ is defined as the bipartite graph, with $m$ vertices on the left (corresponding to the list of codewords) and $n$ vertices on the right (corresponding to the blocks), and an edge connecting $i\in [m]$ on the left with $j \in [n]$ on the right if the encoding of $f_i$ agrees with $y$ at coordinate $j$. 
\end{definition}

All the proofs will essentially attempt to upper-bound the number of edges in any agreement graph, and thereby infer that there cannot be ``too many'' left-vertices (codewords) with ``large degree'' (agreement). 

\section{KRSW/Tamo's upper bound for list size}\label{sec:KRSW}

The bounds due to Kopparty, Ron-Zewi, Saraf and Wootters work for any linear code, not necessarily FRS codes and we will also state the results in that generality. Let $\bF$ be any finite field, $s$ a positive integer and $\cC \subseteq (\bF^s)^n$ be an $\bF$-linear code over the alphabet $\bF^s$ with block length $n$ and fractional distance $\delta$.

Let $y \in \br{\bF_q^s}^n$ be an arbitrary point in the code space of $\cC$.
Let $\cL := \bc{f_{1}, \dots, f_{t}} = B(y, \rho) \cap \cC$ be the list of codewords at distance at most $\rho = \delta - \varepsilon$ from $y$.

\begin{definition}[Certificates]
  Let $\cA$ be an affine subspace of an $\bF$-linear code $\cC \subseteq (\bF^s)^n$.
  A \emph{certificate with respect to $y$} is a sequence of coordinates $(i_1,\ldots, i_a)$ (each $i_j \in [n]$) such that there is a unique codeword $f\in \cA$ that agrees with $y$ at the coordinates $i_1,\ldots, i_a$;  we shall say that this is a \emph{certificate for $f$}. 

  We shall call such a certificate a \emph{minimal certificate} if $i_1,\ldots, i_{a'}$ is \emph{not} a certificate for any $a' < a$. 
\end{definition}

Equivalently, if $G = G(\cL, y)$ is the agreement graph, then a certificate for $f$ identifies a set of right vertices with unique common neighbour being $f$.

\begin{theorem}[\cite{KoppartyRSW2023,Tamo2024}]
  Let $\cC \subseteq (\bF^s)^n$ be an $\bF$-linear code with distance $\delta$ and $\rho = \delta - \epsilon$ and $y \in \br{\bF_q^s}^n$ an arbitrary message.
  Suppose $\cL := B(y,\rho) \cap \cC$ is contained in an affine space of dimension $r$. 

  Then there is a probability distribution on minimal certificates with respect to $y$ such that for any $f \in \cL$, the set of of minimal certificates for $f$ of length at most $r$ has probability mass at least $\epsilon^r$

  In particular, the size of $\cL$ is upper-bounded by $(1/\epsilon^r)$. 
\end{theorem}
\begin{proof}
  Let $\cA$ be the affine subspace of dimension $r$ containing $\cL$.
  The distribution on minimal certificates is the most natural one --- start with $C_0 = \emptyset$ and extend it by choosing a uniformly random coordinate, one coordinate at a time, until it becomes a minimal certificate. Fix any $f \in \cL$ for the rest of the argument. The goal is to show that the probability mass on short certificates for $f$ is large.
  \noindent
  For a set of coordinates $S = \set{i_1,\ldots, i_t}$, we will define $\cA(S)$ as 
  \[
  \cA(S) : = \setdef{f\in \cA}{\text{$f$ agrees with $y$ at coordinate $i$, for all $i \in S$}}\;.
  \]
  To begin with, $\cA^{(0)} := \cA$ contains $f$. Assume that we have constructed a partial certificate $C_j = (i_1,\ldots, i_j)$ so far with $\cA^{(j)} := \cA(C_j)$ being an affine space containing $f$. 
  Whenever we have $\cA^{(j)} = \cA(C_j) \neq \{f\}$ (i.e., $C_j$ is not yet a certificate for $f$), let $f'\neq f$ be any other element of $\cA^{(j)}$. Since $f'$ and $f$ are distinct codewords, they agree on at most $(1-\delta)n$ coordinates but $f$ agrees with $y$ on more than $(1-\delta + \epsilon)n$ coordinates. Hence, if $i_{j+1}$ was chosen to be any of the coordinates where that $f$ agrees with $y$ but disagrees with $f'$ on that coordinate, then we have that $\cA^{(j+1)} := \cA(\set{i_1,\dots, i_{j+1}})$ continues to contain $f$ but is a strictly smaller subspace of $\cA^{(j)}$. Thus, with probability at least $\epsilon$ on the choice of $i_{j+1} \in [n]$, we have that for $C_{j+1} = (i_1,\ldots, i_{j+1})$
  \begin{align*}
    f\in \cA^{(j+1)} \text{ and }\dim\cA^{(j+1)} < \dim \cA^{(j)}. 
  \end{align*}
  where $\cA^{(j+1)} = \cA(C_{j+1})$. 
  Since $\dim\cA^{(0)} \leq r$, with probability at least $\epsilon^r$ we get a minimal certificate for $f$ of length at most $r$. \qedhere

\end{proof}


\section{Dimension of typical subspaces obtained from restrictions}\label{sec:GK}

In the above proof, we started with an $r$-dimensional space $\cA$ that included all our codewords of interest, and we considered various subspaces $\cA_i$ defined as
\[
\cA_i := \setdef{f\in \mathcal{A}}{f\text{ agrees with $y$ at coordinate $i$}}
\]
and let $r_i := \dim \cA_i$. In the above proof, we mainly used the fact that $r_i < r$ for at least $\epsilon n$ many choices of $i$. The following lemma of Guruswami and Kopparty~\cite{GuruswamiK2016} says that, for FRS codes, the average $r_i$ is significantly smaller than $r$. 

\begin{restatable}[Guruswami and Kopparty~\cite{GuruswamiK2016}]{lemma}{guruswamikopparty}
  \label{thm:kopparty-guruswami}
  Let $y \in \br{\bF_q^s}^n$ and $\cA$ be an affine subspace of $\FRS_{k,s}$ of dimension $r$. For each $i \in [n]$, define
  \[
  \cA_i =  \setdef{f\in \mathcal{A}}{f\text{ agrees with $y$ at coordinate $i$}}
  \]
  with $r_i = \dim \cA_i$. Then, 
  \[
  \sum_{i\in [n]} r_i \leq r\cdot \tau_r \cdot n
  \]
  for $\tau_r = \frac{s R}{s - r + 1}$ where $R = k/ns$. 
\end{restatable}

\noindent
In other words, $\E_i[r_i] \approx r \cdot R$. More precisely, if $s = \Theta(1/\epsilon^2)$ and $r = \Theta(1/\epsilon)$, then $\tau_r = R \cdot \br{1 + \Theta(\epsilon)}$.  \\

For the sake of completeness, we add a proof of the above lemma in \cref{sec:proof-of-guruswami-kopparty}. 
Using this lemma, we can obtain significantly better bounds on the list size for FRS codes. 

\section{Srivastava's improved list size bound}\label{sec:Sri}

Srivastava's~\cite{Srivastava2025} main theorem is the following.
\begin{theorem}[Better list size bounds for FRS codes \cite{Srivastava2025}]
  \label{thm:shashank-list-size-bound}
  If $y$ is any received word, and $\cA$ is an affine subspace of dimension $r$, then for any $r \leq t \leq s$ we have $$\abs{B\br{y, \frac{t}{t+1}\br{1-\frac{s}{s-r+1}R}} \cap \cA} \leq (t-1)r + 1.$$
  \noindent Writing in terms of $\tau_r = \frac{sR}{s - r + 1}$ (as in \cref{thm:kopparty-guruswami}), the above can be written as
  \[
    \abs{B\br{y, \frac{t}{t+1}\br{1-\tau_r}} \cap \cA} \leq (t-1)r + 1\;.
  \]
\end{theorem}

\paragraph{Setting parameters:} For any parameter $\epsilon > 0$, we can set $t = 2/\epsilon$, and $s = 3/\epsilon^2$. By Guruswami and Wang~\cite{GuruswamiW2013}, we know that 
\[
B\br{y, \frac{t}{t+1}(1 - \tau_{t})} \cap \FRS_{k,s}
\]
is contained in an affine subspace of dimension at most $r = t-1$. Plugging these parameters in, we can check that $\rho = \frac{t}{t+1}(1 - \tau_t) \geq 1 - R - \epsilon$. 
\begin{align*}
  \rho  & = \frac{t}{t+1} \cdot \br {1 - \frac{sR}{s - t + 2}}\\
  & = \frac{(2/\epsilon)}{(2/\epsilon + 1)}\br {1 - R \cdot \frac{3/\epsilon^2}{3/\epsilon^2  - 1/\epsilon + 2}}\\
  & = \frac{1}{(1 + \frac{\epsilon}{2})}\br {1 - R \cdot \frac{1}{1 - (\epsilon/3) + (2/3)\epsilon^2}}\\
  & = (1 - \frac{\epsilon}{2} \pm \Theta(\epsilon^2)) \cdot \br{1 - R\br{1 + \frac{\epsilon}{3} \pm \Theta(\epsilon^2)}}\\
  & = 1 - R - \epsilon\br{\frac{1}{2} + \frac{R}{3}} \pm \Theta(\epsilon^2)\\
  & \geq 1 - R - \epsilon\;.
\end{align*}
In that case, we get that $\FRS_{k,s}$ codes are $(1-R-\epsilon, O(1/\epsilon^2))$-list-decodable. \\

\subsection{Proof of \cref{thm:shashank-list-size-bound}}

The above theorem is proved by induction on the dimension $r$.
The base case is when $r = 1$.
The following bound holds for any linear code.

\begin{lemma}(\cref{thm:shashank-list-size-bound} for $r = 1$)
  If $y$ is any received word, and $\cA$ is an affine subspace of dimension $1$, then for any $t \geq 1$ we have $$\abs{B\br{y, \frac{t}{t+1}(1-R)} \cap \cA} \leq t\;.$$
\end{lemma}

  


\begin{proof}

   Let $L = \abs{B\br{y, \frac{t}{t+1}(1-R)} \cap \cA}$ 

   Suppose the affine space $\cA$ is of form $\setdef{f_0 + \sigma f_1}{\sigma \in \F_q}$, with $f_1 \neq 0$. Let us use $S : = \setdef{i \in [n]}{(\FRS_{k,s}(f_1))_i \neq 0}$ to denote the support of the encoding of $f_1$. Note that $\abs{S} \geq (1 - R)\cdot n$. 

   Notice also that two distinct codewords in the list have to disagree completely on $S$. Hence, every right-side vertex in $S$ has at most one outgoing edge. 

   We now count edges in the agreement graph, from both sides. From the codewords side, since each codeword has agreement strictly more than $n(1-\frac{t}{t+1}(1-R))$, the number of edges is more than $L n(1-\frac{t}{t+1}(1-R))$.

   From the locations side, each vertex in $S$ contributes at most one edge. Each vertex outside $S$ may contribute up to $L$ edges. This gives the total number of edges to be at most $|S| + L(n-|S|) = Ln - (L - 1)|S| \leq Ln - (L - 1)(1-R)n $ using the above lower bound on $|S|$. 
  
   Combining the upper and lower bound on the number of edges,
   \[  L n\inparen{1-\frac{t}{t+1}(1-R)} < Ln - (L - 1)(1-R)n  \]
   Rearranging and cancelling out $Ln$ on both sides,
   \[ (L-1)(1-R)n < L \frac{t}{t+1} (1-R)n \]
 Solving for $L$ gives $L < t+1$.
\end{proof}

We now prove the main theorem.

\begin{proof}[Proof of \cref{thm:shashank-list-size-bound}]
  Let $\rho := \frac{t}{t+1}\br{1-\tau_r}$ and we wish to bound the size of $B(y,\rho) \cap \cA$. 
  \[
  L(r) := \max_{\cA\;:\;\dim\cA = r} \abs{B(y,\rho) \cap \cA}.
  \]
  We will prove a bound on $L(r)$ by inducting on $r$. 
  \begin{quote}
    Inductive claim: $L_i \leq L(r_i) \leq \sigma \cdot r_i + 1$ for a constant $\sigma$ independent of $r_i$.
  \end{quote}
  We will eventually show $\sigma = t-1$ would be sufficient, giving us the requisite bound. \\

  For each $i$, let $\cA_{i}$ be the subspace of $\cA$ corresponding to agreement at coordinate $i$ with $y$.
  Let $r_{i} := \dim \cA_{i}$.
  Let $L$ be the number of codewords in $B(y, \rho) \cap \cA$, and let $L_{i}$ be the number of codewords in $B(y, \rho) \cap \cA_{i}$.
  By the induction hypothesis, for every $i$ such that $r_{i} < r$, we have  $L_{i} \leq L(r_i) \leq \sigma r_{i} + 1$. 

  We count the number of edges in the agreement graph.
  Counting from the left, each codeword has agreement at least $(1 - \rho)n$, therefore the number of edges is at least $(1 - \rho)nL$.
  Counting from the right, coordinate $i$ is incident to at most $L_{i}$ codewords, therefore the number of edges is at most $\sum_{i} L_{i}$.
  Combining this, we have the inequality $\sum_{i} L_{i} \geq (1 - \rho)nL$.

  We cannot use induction to control the coordinates where $r_{i} = r$, therefore for these coordinates we use the trivial bound $L_{i} \leq L$.
  Let $\cB$ be the set of coordinates for which this is true.
  We therefore have
  $$
  \sum_{i \not \in \cB} \br{\sigma \cdot r_{i} + 1} \geq L\br{(1 - \rho)n - \abs{\cB}}\;.
  $$
  Every codeword in the list agrees with $y$ on the set $\cB$, therefore in particular the codewords agree with each other on this set.
  Since any two codewords can have agreement at most $Rn$, we have $\abs{\cB} \leq Rn$, which implies the term $((1 - \rho)n - \abs{\cB})$ is positive.
  Therefore, we can deduce 
  $$
  L \leq \frac{\sum_{i \not \in \cB} \br{\sigma \cdot r_{i} + 1}}{(1 - \rho)n - \abs{\cB}}\;.
  $$

  By \cref{thm:kopparty-guruswami}, we have
  \begin{align*}
    \sum_{i\in [n]} r_i  & = \sum_{i\notin \cB} r_i + \abs{\cB}\cdot r \leq rn \tau_r\;\\
    \implies \sum_{i \notin \cB} (\sigma \cdot r_i + 1) & \leq \sigma \cdot rn \tau_r - \sigma \cdot r \abs{\cB} + (n - \abs{\cB})\\
    & = \sigma \cdot rn \tau_r + n - \abs{\cB}(\sigma \cdot r + 1)\;.\\
    \implies L & \leq \frac{\sigma rn \tau_r + n - \abs{\cB}(\sigma r + 1)}{(1 - \rho)n - \abs{\cB}}\;.
  \end{align*}

  To complete the induction, we have to show $L \leq \sigma \cdot r + 1$.
  From the above, it suffices to show
  \begin{align*}
  0 &\leq \br{(1 - \rho)n - \abs{\cB}} \cdot \br{\sigma r + 1} - \br{\sigma rn \tau_r + n - \abs{\cB}(\sigma r + 1)}\\
  & = (1 - \rho)n \cdot \br{\sigma r + 1} - \br{\sigma rn \tau_r + n}\;.
  \end{align*}
  Indeed, using the fact that $\rho = \frac{t}{t+1} \cdot (1 - \tau_r)$, we have
  \begin{align*}
    (1 - \rho) \cdot \br{\sigma r + 1} - \br{\sigma r \cdot \tau_r + 1} & = \sigma r \cdot ((1-\rho) - \tau_r) + (1 - \rho - 1)\\
    & = \sigma r \cdot ((1-\tau_r) - \rho) - \rho\\
    & = \sigma r \cdot \rho \cdot \inparen{\frac{t+1}{t} - 1} - \rho\\
    & = \rho \cdot \inparen{\frac{\sigma r}{t} - 1} \geq 0 \quad \text{for $\sigma = (t-1)$}
  \end{align*}
  since $(t-1)r \geq t$ as $t > r \geq 2$. 
\end{proof}

From the above proof, it feels like we could have perhaps taken $\sigma = \frac{t}{r}$, thereby getting a list size bound of $L(r) \leq \sigma r + 1 \leq t+1$ instead of $O(tr)$. However, note that the above proof used the fact that $\sigma$ was independent of $r$ (when we bounded $\sum_{r_i < r} L(r_i)$ with $\sigma \sum r_i + (n - \abs{\cB})$). Nevertheless, this perhaps suggests that there is some slack in the above analysis and one could perhaps improve the analysis to obtain a list-size bound of $O(t)$ instead of $O(tr)$. 

Chen and Zhang~\cite{ChenZ2025} (independent and parallel to Srivastava~\cite{Srivastava2025}) obtain an $O(t)$ bound by using induction to bound the number of edges of the agreement graph rather than bounding the list size directly. 

\section{Further improvements on the list size due to Chen and Zhang}\label{sec:CZ}

\begin{theorem}[Chen and Zhang~\cite{ChenZ2025}]
  \label{thm:chen-zhang}
  Let $y \in \br{\bF_q^s}^n$ be an arbitrary received word for the $\FRS_{k,s}$ code. For any $0\leq t \leq s$, we have
  \[
  \abs{B\br{y, \frac{t}{t+1} \br{1 - \tau_t}} \cap \FRS_{k,s}} \leq t\;,
  \]
  where $\tau_t = \frac{s R}{s - t + 1}$ (as in \cref{thm:kopparty-guruswami}). 
\end{theorem}

\paragraph{Setting parameters:} As in the previous case, if $t = 2/\epsilon$ and $s = 3/\epsilon^2$ we once again have $\rho = \frac{t}{t+1}\br{1 - \tau_t}\geq 1 - R - \epsilon$, the above theorem shows that $\FRS_{k,s}$ codes are 
$(1 - R - \epsilon, 2/\epsilon)$-list-decodable. 

\medskip

\begin{remark*}
Unlike the previous bound of Srivastava (\cref{thm:shashank-list-size-bound}), the above bound is oblivious of any ambient space that the codewords lie in. In particular, the above list-size bound does not rely on the fact from Guruswami and Wang~\cite{GuruswamiW2013} that all close-enough codewords lie in a low-dimensional affine space. 
\end{remark*}

The above theorem will be proved by once again considering relevant agreement graphs and upper-bounding the number of edges in it. For a set of distinct codewords $\set{f_1,\ldots, f_m}$ and a received word $y \in \br{\bF_q^s}^n$, let $G = G(\set{f_1,\ldots, f_m},y)$ be the agreement  graph. We will use $E_G$ to denote the number of edges in $G$. For any subset $H$ of left vertices in $G$, let $E_H$ denote the number of edges in the induced graph $G(H, y)$. 

Let $n_G$ be the number of right vertices of $G$ that have degree at least $1$ (these are the positions where at least one of the codewords agrees with $y$). Similarly, for any subgraph induced by a set $H$ of left vertices, $n_{H}$ is the number of right vertices with degree at least $1$ (we overload notation and use $H$ for both the subset of vertices and the induced subgraph).

The main technical lemma of Chen and Zhang can be stated as follows. 

\begin{lemma}[Chen and Zhang~\cite{ChenZ2025}]
  \label{lem:chen-zhang}
  Let $\mathcal{A}$ be the affine subspace spanned by $f_1, \dots, f_m$ and suppose $r$ be the dimension of this affine space. Then, for any agreement graph $G = G(\set{f_1,\dots, f_m}, y)$ corresponding to a message $y \in \br{\bF_q^s}^n$, we have
  \[
  E_G \leq \frac{(m-1)k}{s - r + 1} + n_{G}\;.
  \]
\end{lemma}

\noindent
Recalling the parameter $\tau_r$ from \cref{thm:kopparty-guruswami}, the above can be restated as saying
\[
E_G \leq (m-1) \cdot n \cdot \tau_r + n_G\;.
\]

\noindent 
Before we see the proof of the above lemma, let us see how \cref{lem:chen-zhang} implies \cref{thm:chen-zhang}. 

\begin{proof}[Proof of \cref{thm:chen-zhang}]
Assume on the contrary that there are $t+1$ distinct codewords $f_1,\ldots, f_{t+1}$ that with fractional distance less than $\rho$ from $y$, where $\rho = \frac{t}{t+1} (1 - \tau_t) = \frac{t}{t+1} - \frac{t}{t+1}\tau_t$. Consider the agreement graph $G = G(\set{f_1,\ldots, f_{t+1}}, y)$. By counting edges from the left, we have that 
\[
\abs{E_G} > (1 - \rho)n \cdot (t+1) = (t\tau_t + 1)n\;.
\]
On the other hand, note that any set of $t+1$ codewords is contained in an affine space of dimension $r \leq t$. Thus, using \cref{lem:chen-zhang} we have
\[
\abs{E_G} \leq (t+1 - 1) \cdot n \cdot \tau_t  + n_G \leq (t \tau_t + 1) \cdot n
\]
contradicting the above bound. Hence the size of the list must be at most $t$. 
\end{proof}

\subsection{Proof of \cref{lem:chen-zhang}}

  Recall that we have to prove that for $G = G(\set{f_1,\ldots, f_m}, y)$, the number of edges $\abs{E_G}$ is upper-bounded by 
  \[
  E_G \leq (m - 1) \cdot n \cdot \tau_r  + n_G\;.
  \]
  where $r$ is the dimension of the smallest affine space $\cA$ containing $f_1,\ldots, f_m$. \\

  \noindent
  The proof is by induction on $m$. The case of $m = 1$ is trivial, since $E_G = n_G$ when $m = 1$.


  Now suppose $m \geq 2$. Hence $r \geq 1$.
  We partition the set of codewords as follows.
  Let $f^{(0)}, f^{(1)},\dots,f^{(r)}$ be $r+1$ codewords in the list $\{f_1,\ldots,f_m\}$ such that the smallest affine space generated by $\set{f^{(0)},\ldots,f^{(r)}}$ is $\cA$. For $i=0, \ldots, r$, let $\cA^{(i)}$ be the smallest affine space generated by $\set{f^{(0)},\ldots,f^{(i)}}$. Observe that the affine dimension of $\cA^{(i)}$ is $i$ and $f^{(i)} \in \cA^{(i)} \setminus \cA^{(i-1)}$ where we have defined $\cA^{(-1)} := \emptyset$. For $i=0,\ldots, r$, define 
  \begin{align*}
    H'_i & := \cA^{(i)} \cap \set{f_1,\ldots, f_m}\,,\\
    H_i & := H'_i \setminus \cA^{(i-1)}\,.
  \end{align*}
  Clearly $(H_0,\ldots,H_r)$ is a partition of $\set{f_1,\ldots, f_m}$. Furthermore, $f^{(i)}\in H_i$ and hence $H_i \neq \emptyset$. 
  Let $m_{i} := \abs{H_{i}}$.
  We have $\sum m_i = m$ and each $0 \neq m_{i} < m$ since $m_0 =1$ and $r\geq 1$. Let $r^{(i)}$ be the affine dimension of $H_i$. 

  We apply the inductive hypothesis on the subgraphs induced by $H_{0}, \dots, H_{r}$.
  The induced subgraphs are exactly the agreement graphs of the list of codewords in $H_{i}$.
  Therefore by induction we have
  \[
    E_{H_{i}} \leq (m_i - 1)\cdot n \cdot \tau_{r^{(i)}} + n_{H_{i}} \leq (m_i - 1)\cdot n \cdot \tau_r + n_{H_{i}}.
  \]
  The total number of edges of $G$ is the sum of the number of edges in each  induced graph, therefore
  \begin{align*}
    E_{G} = \sum_{i=0}^r E_{H_i} & \leq \sum_{i=0}^r \br{(m_i - 1) \cdot n\cdot \tau_r + n_{H_{i}}} = m\cdot n \cdot \tau_r - (r+1)\cdot n \cdot \tau_r + \sum_i n_{H_i}\\
    & = (m-1) \cdot n \cdot \tau_r - r \cdot n \cdot \tau_r + \sum_{i} n_{H_{i}}\;.
  \end{align*}
  We now relate the quantities $\sum n_{H_{i}}$ with $n_{G}$.

  Consider any right vertex $j\in [n]$ of $G$.
  If $j$ has degree $0$ in $G$, then $j$ does not contribute to $n_{G}$ or to $n_{H_{i}}$ for any $i$. 

  For each $j \in [n]$ with degree at least $1$ in $G$, let $t_j$ be the number of $i$'s such that there is an edge from $j$ to $H_i$. Then $j$ contributes $t_j$ to $\sum n_{H_i}$ and $1$ to $n_G$. Hence, we have $\sum_{i} n_{H_{i}} - n_{G} = \sum_{j} (t_{j} - 1)$. 

  Using this in the equation above gives
  \begin{align*}
    E_{G} &\leq (m-1) \cdot n \cdot \tau_r + n_{G} - r \cdot n \cdot \tau_r + \sum_{j} (t_{j} - 1)\;.
  \end{align*}

  For each such $j \in [n]$, let $\mathcal{A}_j$ be the affine subspace of $\mathcal{A}$ containing all codewords that agree with the message $y$ at coordinate $j$, and let $r_j$ be its dimension. Note by our construction of the partition $H_0, \ldots, H_r$ that any set of vectors chosen by picking at most one from each $H_i$ are affine independent. Hence, since $j$ has edges to $t_j$ different $H_i$'s, we have that $r_j \geq (t_j - 1)$. 

  By \cref{thm:kopparty-guruswami}, we have $\sum (t_j - 1) \leq \sum r_j \leq r \cdot n \cdot \tau_r$. Combining this with the above equation, we have
  \[
  E_G \leq (m-1)\cdot n \cdot \tau_r + n_G\;.
  \]
  That completes the proof of \cref{lem:chen-zhang}. \hfill \qed

\section{List-size lower bounds for list-recovery}\label{sec:recovery}

Although \cref{thm:chen-zhang} gives optimal bounds for list-decoding of FRS codes, Chen and Zhang also show that an exponential dependence in $\epsilon$ is unavoidable for the question of list-recovery. In this section we give their counter-example. 

Recall that the set of evaluation points for the $\FRS_{k,s}$ code are $\alpha_1,\ldots, \alpha_n$, with $\gamma$ being the generator of $\F_q^*$ used for the folding. For each $i \in [n]$, define the polynomial $Q_i(x)$ defined as
\[
Q_i(x) = (x - \alpha) (x - \gamma \alpha) \cdots (x - \gamma^{s-1}\alpha).
\]
The $i$-th symbol of the $\FRS_{k,s}$ encoding of a polynomial $g$ can equivalently also be thought of as the residue $\inparen{g(x) \bmod Q_i(x)}$. 

\medskip

Define integer parameters $m, p$ such that $m \approx \frac{R}{\epsilon} + 1$ and 
\[
p = \floor{\frac{m \floor{\frac{k-1}{s}}}{m-1}} = \frac{m}{m-1} \cdot \frac{k}{s} - O(1) = n(R + \epsilon) - O(1).
\]
Consider the following set of $m$ polynomials: 
\[
\text{For $i = 1,\ldots, m-1$, \quad} f_i(x) := \prod_{\substack{j\in [p]\\j \neq i \bmod{m}}} Q_i(x).
\]
By the choice of $m$ and $k$, it follows that $\deg f_i \leq (k-1)$ for all $i \in [m]$ since each $f_i$ is a product of at most $\frac{m-1}{m}$ of the $Q_j$'s for $j \in [p]$. 

\begin{lemma}[List-recovery for FRS codes~\cite{ChenZ2025}]
  \label{lem:chen-zhang-list-recovery} Let $B$ be any set of $\ell$ distinct field elements. Consider the set of polynomials
  \[
  \mathcal{G} := \setdef{\beta_1 f_1 + \cdots + \beta_m f_m}{\beta_i \in B}\;.
  \]
  Then, $\abs{\mathcal{G}} = \ell^m$ and, for each $i \in [p]$, we have
  \[
  \abs{ \setdef{\inparen{\FRS_{k,s}(g)}_i}{g\in \mathcal{G}}} \leq \ell\;.
  \]
  (That is, the $\FRS$ encoding of any polynomial in $\mathcal{G}$ takes only one of $\ell$ possible values in the first $p$ coordinates.) 
\end{lemma}

Since $p \approx n(R + \epsilon)$, we have a particular instance of list-recover with each coordinate list-size bounded by $\ell$, with $\ell^{R/\epsilon}$ codewords with fractional agreement of $R + \epsilon$. 

\begin{proof}
  To see that $\abs{\mathcal{G}}$ has size $\ell^m$, we observe that the polynomials $f_1,\ldots, f_m$ are linearly independent. Indeed, if $c_1 f_1 + \cdots + c_m f_{m} = 0$, with $c_1 \neq 0$ (without loss of generality), looking at the equation modulo $Q_1(x)$ yields a nonzero quantity on the left-hand side but zero on the right. 

  As for the second claim, let $g = \beta_1 f_1 + \cdots \beta_m f_m$. Then, observe that $(\FRS_{k,s}(g))_i = g \bmod Q_i(x) = \beta_{i'} \inparen{f_{i'}(x) \bmod Q_i(x)}$ where $i' \in [m]$ is the unique value such that $i' = i \bmod{m}$ (since all other $f_j$'s are divisible by $Q_i$). As $\beta_i$'s come from a set of size at most $\ell$, the $i$-th coordinate of $\FRS_{k,s}(g)$ will be one of the $\ell$ scalings of $\inparen{f_{i'}(x) \bmod Q_i(x)}$.
\end{proof}

  {\small 
  \bibliographystyle{prahladhurl}
  \bibliography{ref}

\begin{thebibliography}{KRSW23}

\bibitem[CZ24]{ChenZ2025}
\textsc{Yeyuan Chen} and \textsc{Zihan Zhang}.
\newblock \emph{Explicit folded {R}eed-{S}olomon and multiplicity codes achieve relaxed generalized singleton bound}, 2024.
\newblock (manuscript).
\newblock \href{http://arxiv.org/abs/2408.15925}{\path{arXiv:2408.15925}}.

\bibitem[GK16]{GuruswamiK2016}
\textsc{Venkatesan Guruswami} and \textsc{Swastik Kopparty}.
\newblock \href{https://doi.org/10.1007/s00493-014-3169-1} {\emph{Explicit subspace designs}}.
\newblock Comb., 36(2):161--185, 2016.
\newblock (Preliminary version in \emph{54th FOCS}, 2013).
\newblock \href{https://eccc.weizmann.ac.il/eccc-reports/2013/060}{\path{eccc:2013/060}}.

\bibitem[GR08]{GuruswamiR2008}
\textsc{Venkatesan Guruswami} and \textsc{Atri Rudra}.
\newblock \href{https://doi.org/10.1109/TIT.2007.911222} {\emph{Explicit codes achieving list decoding capacity: Error-correction with optimal redundancy}}.
\newblock IEEE Trans.\ Inform.\ Theory, 54(1):135--150, 2008.
\newblock (Preliminary version in \emph{38th STOC}, 2006).
\newblock \href{http://arxiv.org/abs/cs/0511072}{\path{arXiv:cs/0511072}}, \href{https://eccc.weizmann.ac.il/eccc-reports/2005/133}{\path{eccc:2005/133}}.

\bibitem[GW13]{GuruswamiW2013}
\textsc{Venkatesan Guruswami} and \textsc{Carol Wang}.
\newblock \href{https://doi.org/10.1109/TIT.2013.2246813} {\emph{Linear-algebraic list decoding for variants of {R}eed-{S}olomon codes}}.
\newblock IEEE Trans.\ Inform.\ Theory, 59(6):3257--3268, 2013.
\newblock (Preliminary version in \emph{26th {IEEE} Conference on Computational Complexity}, 2011 and \emph{15th RANDOM}, 2011).
\newblock \href{https://eccc.weizmann.ac.il/eccc-reports/2012/073}{\path{eccc:2012/073}}.

\bibitem[Kra03]{Krachkovsky2003}
\textsc{Victor~Yu. Krachkovsky}.
\newblock \href{https://doi.org/10.1109/TIT.2003.819333} {\emph{{R}eed-{S}olomon codes for correcting phased error bursts}}.
\newblock IEEE Trans.\ Inform.\ Theory, 49(11):2975--2984, 2003.

\bibitem[KRSW23]{KoppartyRSW2023}
\textsc{Swastik Kopparty}, \textsc{Noga Ron{-}Zewi}, \textsc{Shubhangi Saraf}, and \textsc{Mary Wootters}.
\newblock \href{https://doi.org/10.1137/20M1370215} {\emph{Improved list decoding of {F}olded {R}eed-{S}olomon and {M}ultiplicity codes}}.
\newblock SIAM J. Comput., 52(3):794--840, 2023.
\newblock (Preliminary version in \emph{59th FOCS}, 2018).
\newblock \href{http://arxiv.org/abs/1805.01498}{\path{arXiv:1805.01498}}, \href{https://eccc.weizmann.ac.il/eccc-reports/2018/091}{\path{eccc:2018/091}}.

\bibitem[Sri25]{Srivastava2025}
\textsc{Shashank Srivastava}.
\newblock \href{https://doi.org/10.1137/1.9781611978322.64} {\emph{Improved list size for folded {R}eed-{S}olomon codes}}.
\newblock In \textsc{Yossi Azar} and \textsc{Debmalya Panigrahi}, eds., \emph{Proc.\ $36$th Annual {ACM}-{SIAM} Symp.\ on Discrete Algorithms (SODA)}, pages 2040--2050. 2025.
\newblock \href{http://arxiv.org/abs/2410.09031}{\path{arXiv:2410.09031}}.

\bibitem[ST23]{ShangguanT2023}
\textsc{Chong Shangguan} and \textsc{Itzhak Tamo}.
\newblock \href{https://doi.org/10.1137/20M138795X} {\emph{Generalized {S}ingleton bound and list-decoding {R}eed-{S}olomon codes beyond the {J}ohnson radius}}.
\newblock SIAM J. Comput., 52(3):684--717, 2023.
\newblock (Preliminary version in \emph{52nd STOC}, 2020).
\newblock \href{http://arxiv.org/abs/1911.01502}{\path{arXiv:1911.01502}}.

\bibitem[Tam24]{Tamo2024}
\textsc{Itzhak Tamo}.
\newblock \href{https://doi.org/10.1109/TIT.2024.3402171} {\emph{Tighter list-size bounds for list-decoding and recovery of folded {R}eed-{S}olomon and multiplicity codes}}.
\newblock IEEE Trans.\ Inform.\ Theory, 70(12):8659--8668, 2024.
\newblock \href{http://arxiv.org/abs/2312.17097}{\path{arXiv:2312.17097}}.

\end{thebibliography}
}

\appendix

\section{Proof of the Guruswami-Kopparty lemma}
\label{sec:proof-of-guruswami-kopparty}

For the sake of completeness, we give a proof of the \cref{thm:kopparty-guruswami} (restated below):

\guruswamikopparty*

\begin{proof}
  Let the $r$-dimensional affine space $\mathcal{A}$ be $f_0 + \operatorname{\F-span}\inbrace{f_1,\ldots, f_r}$, where $f_1,\ldots, f_r$ are linearly independent polynomials of degree less than $k$. The \emph{Folded-Wronskian}, $W_\gamma(f_1,\ldots, f_r)$ of these polynomials is defined as the following determinant of an $r\times r$ matrix.
  \[
    W_\gamma(f_1,\ldots, f_r) = \abs{
      \begin{array}{cccc}
        f_1(x) & f_2(x) & \cdots & f_r(x)\\
        f_1(\gamma x) & f_2(\gamma x) & \cdots & f_r(\gamma x)\\
        \vdots & \vdots & \ddots & \vdots\\
        f_1(\gamma^{r-1} x) & f_2(\gamma^{r-1} x) & \cdots & f_r(\gamma^{r-1} x)
      \end{array}}
  \]
  We will use $\mathcal{W}_\gamma(f_1,\ldots, f_r)$ to refer to the $r\times r$ matrix above. 
  The above polynomial has degree at most $rk$, and since $f_1,\ldots, f_r$ are linearly independent, it is known that the Folded-Wronskian is a nonzero polynomial. We will relate the $r_i$'s with appropriate roots of $W_\gamma(f_1,\ldots, f_r)$ and their multiplicities. 

  Fix a coordinate $i\in [n]$ and $\alpha_i$ being the correspondent element of $\F$. The space $\mathcal{A}_i$ can be equivalently expressed as all polynomials form $f_0 + \beta_1 f_1 + \cdots + \beta_r f_r$ (where $\beta_1,\ldots, \beta_r \in \F_q$) such that 
  \[
  \beta_1 f_1(\gamma^{j} \alpha_i) + \cdots + \beta_r f_r(\gamma^{j} \alpha_i) = (y_i)_j - f_0(\gamma^j \alpha_i)\quad\text{for $j = 0,\ldots, s-1$}
  \]
  In other words, $\beta_1,\ldots, \beta_r$ are solutions to the linear system
  \[
  \insquare{
    \begin{array}{ccc}
      f_1(\alpha_i) & \cdots & f_r(\alpha_i)\\
      f_1(\gamma \alpha_i) & \cdots & f_r(\gamma \alpha_i)\\
      \vdots & \ddots & \vdots \\
      f_1(\gamma^{s-1} \alpha_i) & \cdots & f_r(\gamma^{s-1} \alpha_i)
    \end{array}
  }_{s\times r}
  \insquare{\begin{array}{c}
    \beta_1\\
    \vdots \\
    \beta_r
  \end{array}
  }_{r \times 1} = 
  \insquare{\begin{array}{c}
    (y_i)_0 - f_0(\alpha_i)\\
    (y_i)_1 - f_0(\gamma \alpha_i)\\
    \vdots \\
    (y_i)_{s-1} - f_0(\gamma^{s-1} \alpha_i)
  \end{array}
  }_{s\times 1}.
  \]
  Hence, if $\dim \mathcal{A}_i = r_i$, then the rank of the $s\times r$ matrix on the LHS is at most $r - r_i$. Furthermore, note that for any $\sigma \in \set{\alpha_i, \gamma \alpha_i, \dots, \gamma^{s - r}\alpha_i}$, the matrix $\mathcal{W}_\gamma(f_1,\ldots, f_r)\mid_{x = \sigma}$ is an $r\times r$ submatrix of the above $s\times r$ matrix. Since the above $s\times r$ matrix has a rank-deficit of $r_i$, we have that each such $\sigma$ must be a root of $W_\gamma(f_1,\ldots, f_r)$ of multiplicity at least $r_i$. 
  Hence,
  \begin{align*}
  \sum_{i \in [n]} r_i (s - r + 1) &\leq \deg(W_\gamma(f_1,\ldots, f_r))  \leq rk\\
  \implies \sum_{i\in [n]} r_i & \leq \frac{rk}{s - r + 1} = r \cdot \tau_r \cdot n.\qedhere
  \end{align*}
\end{proof}

{\let\thefootnote\relax
\footnotetext{\textcolor{\gitinfonotecolour}{\gitinfonote}
}}

\end{document}